\documentclass[letterpaper,10pt,conference]{ieeeconf}  

\IEEEoverridecommandlockouts                              

\usepackage{graphics} 
\usepackage{epsfig} 
\usepackage{amsmath} 
\usepackage{amssymb}  
 
\usepackage{amsthm}
\usepackage{cite}
\usepackage{color}
\usepackage{enumerate}
\usepackage{graphicx}
\usepackage{mathtools}
\usepackage{cancel}
\usepackage{url}

\newcommand{\naturals}{\mathbb{N}}
\newcommand{\real}{\mathbb{R}}
\newcommand{\realnonneg}{\mathbb{R}_{\ge 0}}
\newcommand{\realpos}{\mathbb{R}_{> 0}}

\newcommand{\map}[3]{#1:#2 \rightarrow #3}

\newcommand{\longthmtitle}[1]{\mbox{}{\textit{(#1):}}}
\newcommand{\setdef}[2]{\{#1 \; | \; #2\}}
\newcommand{\setdefb}[2]{\big\{#1 \; | \; #2\big\}}
\newcommand{\setdefB}[2]{\Big\{#1 \; | \; #2\Big\}}
\newcommand*{\SetSuchThat}[1][]{} 
\newcommand*{\MvertSets}{%
    \renewcommand*\SetSuchThat[1][]{%
        \mathclose{}%
        \nonscript\;##1\vert\penalty\relpenalty\nonscript\;%
        \mathopen{}%
    }%
}
\MvertSets 
\DeclarePairedDelimiterX \Set [2] {\lbrace}{\rbrace}
    {\,#1\SetSuchThat[\delimsize]#2\,}

\newcommand{\Cc}{\mathcal{C}}

\newcommand{\Hc}{\mathcal{H}}
\newcommand{\Kc}{\mathcal{K}}

\newcommand{\defeq}{\triangleq}

\newtheorem{theorem}{Theorem}
\newtheorem{assumption}[theorem]{Assumption}
\newtheorem{lemma}[theorem]{Lemma}

\newtheorem{remark}[theorem]{Remark}

\newtheorem{corollary}[theorem]{Corollary}

\theoremstyle{definition}
\newtheorem{definition}[theorem]{Definition}

\newcommand{\des}{{\operatorname{des}}}
\newcommand{\on}{{\operatorname{on}}}
\newcommand{\off}{{\operatorname{off}}}
\newcommand{\fl}{{\operatorname{FL}}}
\newcommand{\Lie}{\mathcal{L}}

\renewcommand{\baselinestretch}{1}
\allowdisplaybreaks

\DeclareMathOperator*{\argmin}{argmin}

\title{\LARGE \bf Stability and Safety through Event-Triggered Intermittent Control \\ with Application to Spacecraft Orbit Stabilization  }

\author{Pio Ong, Gilbert Bahati, and  Aaron D. Ames
\thanks{This research is supported in part by Raytheon Technologies and the National Science Foundation (CPS Award \#1932091).}
  \thanks{Pio Ong, Gilbert Bahati, and Aaron D. Ames are with the Department of Mechanical and Civil Engineering, California Institute of Technology, Pasadena, CA 91125, USA. {\tt\small \{pioong,gbahati,ames\}@caltech.edu}}%
}

\begin{document}
\maketitle
\begin{abstract}
In systems where the ability to actuate is a scarce resource, e.g., spacecrafts, it is desirable to only apply a given controller in an intermittent manner---with periods where the controller is on and periods where it is off.  Motivated by the event-triggered control paradigm, where state-dependent triggers are utilized in a sample-and-hold context, we generalize this concept to include state triggers where the controller is off thereby creating a framework for \emph{intermittent control}.  
Our approach utilizes certificates---either Lyapunov or barrier functions---to design intermittent trigger laws that guarantee stability or safety; the controller is turned on for the period for which is beneficial with regard to the certificate, and turned off until a performance threshold is reached.  The main result of this paper is that the intermittent controller scheme guarantees (set) stability when Lyapunov functions are utilized, and safety (forward set invariance) in the setting of barrier functions.
As a result, our trigger designs can leverage the intermittent nature of the actuator, and at the same time, achieve the task of stabilization or safety. We further demonstrate the application and benefits of intermittent control in the context of the spacecraft orbit stabilization problem.
\end{abstract}
\section{Introduction}
Modern control systems typically implement continuous-time controllers in a discrete fashion---often using sample-and-hold.  The controller is, therefore, held at a constant value over a fixed interval or, through the modern paradigm of event-triggered control \cite{PT:07,WPMHH-KHJ-PT:12,RP-PT-DN-AA:15}, over variable intervals determined by a triggering law.  Yet in many systems it is not desirable to keep the controller on at all times, especially if there are finite resources for actuation, as this requires control effort to be constantly expended.  A prime example of such systems are spacecrafts which have limited resources (e.g., fuel and power), and so thrusters can only fire in an intermittent fashion.  The goal of this paper is to capture this notion in a principled fashion through the introduction of \emph{event-triggered intermittent control}.

\subsubsection*{Literature Review}
There are, broadly speaking, two exiting approaches taken in the synthesis of intermittent controllers: direct and indirect. The first approach designs the controller that captures the intermittent nature in its implementation.  Typically, intermittent control systems are treated as switched systems, either with time-based switching~\cite{CL-GF-XL:07} or state-based \cite{NY-YS-KK-TN:16,RD-HP:19,QW-YH-GT-MW:16}.  The second approach is an indirect one as it first designs the continuous-time controllers and then implements them in an intermittent fashion. Most works \cite{LY-WM-IJ-PE:93,PJG-LW:09,PJG-IL-ML-HG:11,PJG-IL-HG-ML:14} enable intermittent control by prescribing a constant threshold for the error introduced by the intermittent implementation. However, the analyses and methods rely on the systems being linear.  Through the use of certificates, this paper develops an intermittent control framework for nonlinear systems.   

The approach to intermittent control taken in this paper leverages the event-triggered control framework. In this setting, to enable sample-and-hold implementation of a controller, \cite{PT:07} utilizes ISS Lyapunov functions \cite{EDS:08} for their robustness.  Through monitoring the Lyapunov function and ensuring that its time-derivative remains negative, the task of stabilization can be accomplished. Event-triggered control has been employed for intermittent control in \cite{PJG-LW:10}; however, the scheme relies on using a linear predictor, and the result is limited to linear systems.
The main obstruction of using event-triggered control for intermittent control is that the traditional framework developed by \cite{PT:07} does not allow for the possibility of Lyapunov function to increase, which can happen when the control input is set to zero. To this end, there are trigger frameworks, dynamic triggering~\cite{AG:15} and performance-barrier-based triggering \cite{PO-JC:21-tac}, that can accommodate this behavior. In particular, we build our intermittent control results on the work \cite{PO-JC:21-tac} because it already incorporates the concept of control barrier functions~\cite{ADA-SC-ME-GN-KS-PT:19} in its design---this will enable the synthesis of intermittent controllers in the context of safety.  Other examples of event-triggered control in the context of safety include~\cite{GY-CB-RT:19,AJT-PO-JC-AA:21-csl}. The latter work uses the concept of Input-to-state safe barrier function~\cite{SK-ADA:18} to establish minimum inter-event time. This paper builds on the \cite{AJT-PO-JC-AA:21-csl} with the concepts in~\cite{PO-JC:21-tac} to address the intermittent control problem in the context of safety.

\subsubsection*{Statement of Contribution}
This paper presents a novel intermittent control framework for nonlinear systems leveraging certificate-based event triggers.  
To this end, we take inspiration from event-triggered control and its opportunistic nature in sampling control inputs---extended this concept to include periods where the controller is off.   
The first contribution of this paper is an intermittent trigger scheme that opportunistically determines when to switch the controller between on and off in order to achieve (set) stabilization.
Our design is based on monitoring the value of a given input-to-state stable (ISS) Lyapunov function and guaranteeing that overall, it decreases over time. Our second contribution is the extension of this framework to safety, wherein we develop a trigger scheme that guarantees forward invariance of a desired time-varying set. The idea, analogous to the stabilization case, is based on monitoring a strong input-to-state safe barrier function (sISSf-BF). 

The trigger schemes developed and presented in this paper allow continuous-time controllers to be implemented in an intermittent fashion with formal guarantees of stability (via Lyapunov functions) and safety (via barrier functions).  Thus, to the best of our knowledge, this paper presents the first instance of a certificate-based event triggered control scheme that effectively allows for the implementation of continuous-time controllers for nonlinear systems in an intermittent fashion.  
To demonstrate the utility and benefits of intermittent control, we consider stabilization of a spacecraft around a small asteroid in a tight orbit---the spacecraft orbit stabilization problem. 
We provide a detailed exposition of how the theoretic methods developed in this paper can be used to design intermittent control inputs that stabilize the spacecraft to a tube around the desired orbit. 
Simulations demonstrate the effectiveness of our results.

\subsubsection*{Notation}
We use $\naturals$, $\real$, $\realnonneg$, $\realpos$ to denote natural, real, nonnegative, and positive numbers, respectively. For a vector $x\in\real^n$ and a matrix $A\in\real^{n\times n}$, $\|x\|$ and $\|A\|$ denotes the Euclidean norm and the spectral norm, respectively. A function $\map{\alpha}{[0,a)}{\realnonneg}$ with $a>0$, is of class-$\Kc_\infty$ if $\alpha(0)=0$, $\alpha$ is strictly increasing, and $\lim_{s\rightarrow\infty} \alpha(s) = \infty$. Given a system and an initial condition $x_0$, let $\map{x}{\realnonneg}{\real^n}$ be a solution, a time-varying set $\Cc$ is forward invariant if $x(t)\in\Cc(t)$ for all time whenever $x_0\in\Cc(0)$.

\section{Opportunistic Controller Implementation}
This section reviews the concept of event-triggered control. We consider the nonlinear control system: 
\begin{equation}\label{eq:NL}
\dot x = f(x,u)
\end{equation}
where $x\in\real^n$ is the state and $u\in\real^m$ is the control input.
In practice, a control signal of a state-feedback controller $\map{k}{\real^n}{\real^m}$ is digitally implemented in a sample-and-hold fashion. The control signal is updated by sampling the controller at time $t_i$, and it is held constant until the controller next sampled at $t_{i+1}$. The closed-loop system for the time $t\in[t_i,t_{i+1})$ is then given by:
\begin{equation}\label{eq:sah}
\dot x = f(x,k(x(t_i)))=f(x,k(x+e))
\end{equation}
where $e = x(t_i)-x$ is the error introduced from sample-and-hold implementation. Typically,
in order to minimize the error, the controller is sampled periodically and frequently, in which case we can rely on sampling theory to guarantee the expected system behavior. Alternatively, event-triggered control takes a more opportunistic approach by prescribing state-dependent trigger criteria for when the control signal gets adjusted. Under this approach, each prescribed criterion ensures the satisfaction of a certificate that guarantees a desired system objective. For instance, a criterion may be designed to guarantee that a Lyapunov function is monotonically decreasing along the trajectory.

Event-triggered control relies on system robustness to allow for sample-and-hold error. The key idea behind event-triggered control is to monitor the error, and only remove the error (i.e., by adjusting the controller to the correct value) when the system can no longer accommodate for it. One robustness property often used in event-triggered control is from the \emph{Input-to-State Stability (ISS) Lyapunov function} \cite{EDS:08} for the system~\eqref{eq:sah}. In this paper, we also assume the existence of such a function $\map{V}{\real^n}{\real}$ satisfying:
\begin{subequations}
\label{eq:iss}
\begin{align}
    \underline\alpha(\|x\|)&\leq V(x)\leq \overline \alpha(\|x\|),\\ 
    \underbrace{\left. \frac{\partial V}{\partial x} \right|_{x} f(x,k(x + e))}_{\defeq \Lie_fV(x,e)}&\leq-\alpha(\|x\|)+\gamma(\|e\|)     \label{eq:iss_Lie}
\end{align}
\end{subequations}
with class-$\Kc_\infty$ functions $\underline\alpha$, $\overline\alpha$, $\alpha$, and $\gamma$. With the ISS Lyapunov function, \cite{PT:07} provides the trigger design:
\begin{equation}\label{trigger:Paulo}
    t_{k+1}\defeq \min\setdefb{t>t_i}{-\sigma\alpha(\|x(t)\|)+\gamma(\|e(t)\|)=0}
\end{equation}
with a design parameter $\sigma\in(0,1)$ for robustness. This trigger criterion ensures that $\gamma(\|e\|)<\sigma\alpha(\|x\|)$, and the Lyapunov function is monotonically decreasing along any trajectory. Furthermore, as established in \cite{PT:07} under mild assumptions, there exists a minimum inter-event time (MIET) between two consecutive update times, which rules out the possibility of an incomplete solution, i.e., Zeno behavior.

We discuss next the performance-barrier-based trigger framework~\cite{PO-JC:21-tac}, which builds on the trigger design~\eqref{trigger:Paulo}. The framework incorporates the concept of barrier function~\cite{ADA-XX-JWG-PT:17} to abandon the monotonic decrease of the Lyapunov function in order to extend inter-event times. 
In this framework, we use a time-varying set defined with the sublevel set of the Lyapunov function.

\begin{definition}\longthmtitle{Safe Performance Set}
A \emph{performance specification function} is a continuously differentiable function~$\map{S}{\real}{\realnonneg}$. 
Associated with $S$, and an ISS Lyapunov function $V$, is the time-varying \emph{safe performance set}:  
\begin{eqnarray}
\label{eqn:setCc}
\Cc(t) = \setdefb{x\in \real^n}{V(x)\leq S(t)}. 
\end{eqnarray}
\end{definition}

With the performance specification function, the trigger design put forth by~\cite{PO-JC:21-tac} is given by:
\begin{multline}\label{trigger:pbb}
t_{k+1} \defeq \min\setdefB{t>t_i}{\Lie_fV(x(t),e(t))-(1-\sigma)\alpha(\|x(t)\|)\\=\beta(S(t)-V(x(t)))}
\end{multline}
where $\beta$ is a class-$\Kc_\infty$ function. The condition in the trigger above is based on the control barrier function concept, and it ensures that the system trajectory is contained within the safe performance set~$\Cc$, under an assumption on the time-derivative of the performance specification function~$S$. Under this approach, one notable benefit is that the certificate~$V$ is allowed to increase along the trajectory.
We will exploit this special characteristic in the following section.

\section{Event-Triggered Intermittent Control}

This section proposes a trigger scheme that enables the implementation of a state-feedback controller in an intermittent fashion. We consider systems where the controllers gets turned off after each usage. Particularly, the controller is sample-and-held from time $t_i^\on$ to $t_i^\off$, and it is then turned off up to the time $t_{i+1}^\on$. Assuming that $k(0)=0$, the \emph{intermittently-implemented control system} is given by: 
\begin{equation}\label{eq:intermittent}
\dot x = \begin{cases} f(x,k(x(t_i^\text{on}))), &~t\in [t_i^\text{on},t_i^\text{off}) \\
 f(x,k(0)), &~t\in [t_i^\text{off},t_{i+1}^\text{on}) \end{cases}
\end{equation}
where $t_i^\on <t_i^\off$ for all $i\in\naturals$. Our goal is to design trigger conditions that prescribe the sequences $\{t_i^\on\}_{i\in\naturals}$ and $\{t_i^\off\}_{i\in\naturals}$ such that a desired set (or just the origin) is asymptotically stable.

Figure \ref{fig: demo} shows an overview of our trigger scheme. Our main idea is to monitor a certificate function as it evolves in time, more specifically, we ensure that $V$ is strictly decreasing when the controller is on, but only ensure that it is below a performance specification function~$S$ when the controller is off. The latter allows for $V$ to increase, but we can make sure it decreases overall by properly prescribing the performance specification~$S$.

\subsection{Trigger Design: Controller Off}

We start with the trigger design for when to turn the controller off. We follow the simple idea of allowing the controller usage as long as it performs desirably or until it is required to be turned off for other reasons. That is, we define the trigger:
\begin{subequations}\label{trigger:off}
\begin{multline}\label{trigger:off_deriv}
  t_i^* \defeq  \min\setdefb{t>t_i^\on}{\\\Lie_fV(x(t),e_i^\on(t))+(1-\sigma)\alpha(\|x(t)\|)=0}
\end{multline}
where $e_i^\on(t)= x(t_i^\on)-x(t)$. Then with $T_{\max} \in \realpos \cup \infty$ as the maximum time the controller can remain on, we turn the controller off at:
\begin{equation}
  t_i^\off \defeq \min \{t_i^*,t_i^\on+T_{\max}\}.  
\end{equation}
\end{subequations}
By design, the trigger assures $\frac{dV}{dt}<(\sigma-1)\alpha(\|x(t)\|)$  between time $t_i^\on$ and $t_i^\off$, and the certificate~$V$ decreases during the time interval. In addition, the work~\cite{PT:07} provides a MIET for the trigger design~\eqref{trigger:off_deriv}, which we can use to establish a MIET for our trigger design~\eqref{trigger:off}, as stated below. 
\begin{lemma}\longthmtitle{Minimum Inter-event Time}
\label{lem:miet}
Consider the intermittently-implemented control system~\eqref{eq:intermittent}. Given an ISS Lyapunov function~\eqref{eq:iss} for the sample-and-hold system \eqref{eq:sah}, let $t_i^\off$ be determined according to the trigger~\eqref{trigger:off}. Assume $f$, $k$, $\alpha^{-1}$, and $\gamma$ are locally Lipschitz. Then, for a given forward invariant compact set $\Omega$, if $x(t_i^\on)\in\Omega$, then there exists a MIET $\tau$ such that $t_i^\off-t_i^\on\geq \tau$.
\end{lemma}
\begin{proof}
We begin by noting that at time $t_i^\on$, $e_i^\on(t_i^\on) = 0$ by its definition. Then, the left hand side of the trigger condition~\eqref{trigger:off_deriv} is bounded as:
$$
\Lie_fV(x(t_i^\on),0)+(1-\sigma)\alpha(\|(x(t_i^\on)\|) \leq -\sigma\alpha(\|(x(t_i^\on)\|)
$$
from the bound~\eqref{eq:iss_Lie}. Hence, it is negative at the beginning of the interval $[t_i^\on,t_i^\off)$. As the error $e_i^\on$ grows, the trigger condition~\eqref{trigger:off_deriv} monitors and ensures that the left hand side expression remains negative for the duration. Here we note from the bound~\eqref{eq:iss_Lie}:
$$
\Lie_fV(x,e_i^\on)+(1-\sigma)\alpha(\|x\|) \leq -\sigma\alpha(\|x\|)+\gamma(\|e_i^\on\|).
$$
Along the trajectory, the upper bound must reach the value of zero first, and thus:
\begin{align*}
t_i^*  \geq \min\setdefb{t>t_i^\on}{-\sigma\alpha(\|x(t)\|)+\gamma(\|e_i^\on(t)\|)=0}.
\end{align*}
Then, as proven in~\cite[Thm. 3.1]{PT:07}, under the Lipschitzness assumptions, there exists a MIET $\tau^*$ such that the right hand side is lower bounded by $t_i^\on+\tau^*$. Hence, $t_i^*-t_i^\on\geq \tau^*$, and the MIET for the overall trigger design~\eqref{trigger:off} is $\tau = \min\{\tau^*,T_{\max}\}$.
\end{proof}
\vspace{.2cm}

The MIET provided in Lemma~\ref{lem:miet} rules out the possibility of a Zeno solution to the intermittently-implemented control system~\eqref{eq:intermittent}. Thus, if we can enforce the convergence of the certificate $V$ towards zero, we may conclude asymptotic stability of the equilibrium of the system. To this end, we next discuss our trigger design for when to turn the controller back on, in order to guarantee an overall decrease of the certificate $V$ over time.

\begin{figure}[t!]
\centering
\includegraphics[width=\linewidth,height=\textheight,keepaspectratio]{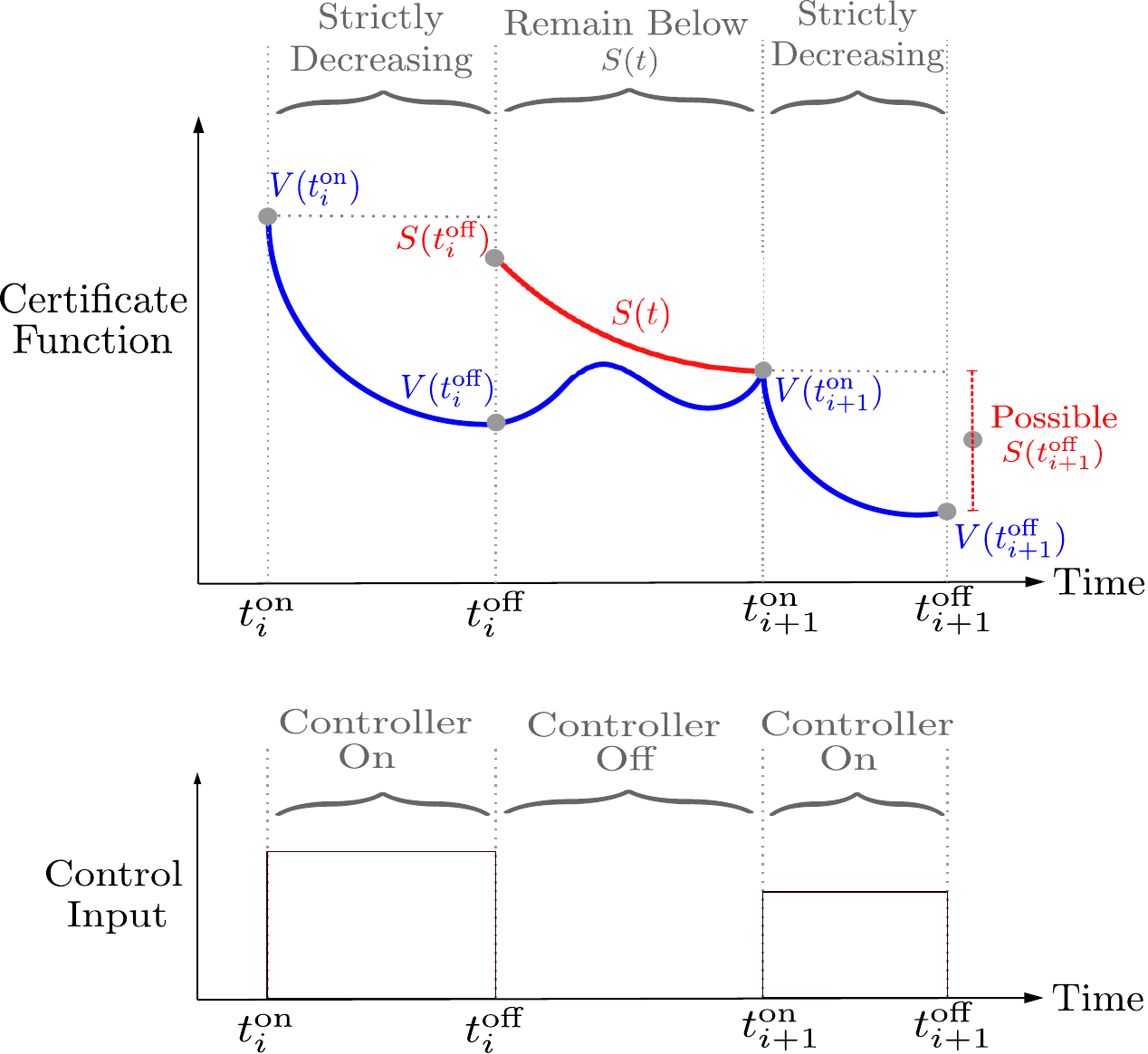}
\caption{Trigger scheme overview for intermittent controller implementation. }
\label{fig: demo}
\end{figure}

\subsection{Trigger Design: Controller On}

The certificate $V$ may increase while the controller is off. When the controller is turned off, the control input is set to zero. The sample-and-hold-error $e_i^\on$ is replaced with $e_i^\off$. Assuming $k(0)=0$, we have $k(0) = k(x+e_i^\off)$. Thus, we find that $e_i^\off=-x$ and:
$$
\frac{dV}{dt}(t) = \Lie_fV(x(t),e_i^\off(t)) \leq -\alpha(\|x(t)\|)+\gamma(\|x(t)\|).
$$
In general, we cannot assume that the origin is asymptotically stable without any control input, so it is possible that the rate of change of $V$ is positive. As such, we leverage the performance-barrier-based trigger design~\eqref{trigger:pbb} idea of allowing the certificate $V$ to increase. We summarize our trigger design in the following result.
\begin{lemma}\longthmtitle{Trigger Design: Controller On}
\label{lem:on}
Consider the intermittently-implemented control system~\eqref{eq:intermittent}. Given an ISS Lyapunov function~\eqref{eq:iss} for the sample-and-hold system \eqref{eq:sah}, let $t_{i+1}^\on$ be determined according to:
\begin{multline}\label{trigger:on}
t_{i+1}^\on \defeq \min\setdefb{t>t_i^\off}{\Lie_fV(x(t),e_i^\off(t))-\frac{dS}{dt}(t)\\=c_{\beta,i}(S(t)-V(x(t)))}
 \end{multline}
with a design parameter $c_{\beta,i}>0$. Assuming $V(x(t_i^\off))\neq S(t_i^\off)$, let $c_{\beta,i}$ be such that:
\begin{equation}\label{eq:beta_param}
c_{\beta,i} > \frac{\Lie_fV(x(t_i^\off),e_i^\off(t_i^\off))-\frac{dS}{dt}(t_i^\off)}{S(t_i^\off)-V(x(t_i^\off))}. 
\end{equation}
Then for time $t\in[t_i^\off,t_{i+1}^\on)$, $S(t)>V(x(t))$ and the safe performance set $\Cc$ is rendered forward invariant.
\end{lemma}
\begin{proof}
Given $x(t_i^\off)\in \Cc(t_i^\off)$ and $V(x(t_i^\off))\neq S(t_i^\off)$, we have $V(x(t_i^\off))<S(t_i^\off)$. Thus, $c_{\beta,i}$ satisfying~\eqref{eq:beta_param} ensures:
$$
\Lie_fV(x(t),e_i^\off(t))-\frac{dS}{dt}(t)< c_{\beta,i}\big(S(t)-V(x(t))\big)
$$
at time $t=t_i^\off$. Then because of the trigger design~\eqref{trigger:on}, the inequality above continues to hold, due to continuity of the trajectory, until it becomes equality at $t_{i+1}^\on$. As a result, it follows from the Comparison Lemma \cite[Lemma 3.4]{HKK:02} that: 
\begin{align*}
\frac{d(V\circ x)}{dt}(t)-\frac{dS}{dt}(t) & ~ < ~   c_{\beta,i}(S(t)-V(x(t))), \nonumber\\
\Rightarrow ~ S(t)-V(x(t)) & ~ > ~  e^{-c_{\beta,i} (t -t_i^\off) } (S(t_i^\off) - V(x(t_i^\off)))  \nonumber\\
\Rightarrow ~ S(t)-V(x(t)) & ~ > ~   0, \qquad \qquad \forall ~ t \in [t_i^\off,t_{i+1}^\on), \nonumber
\end{align*}
since $S(t_i^\off) - V(x(t_i^\off)) > 0$.   Therefore, the performance set $\Cc$ is forward invariant. 
\end{proof}
\vspace{.2cm}

Lemma~\ref{lem:on} provides a trigger for when the controller needs to be back on. The trigger has two design elements: the constant $c_{\beta,i}$ and the performance specification function~$S$. Regarding the former, it may be difficult to find a common $c_\beta$ that satisfy the bound~\eqref{eq:beta_param} for all $i\in\real^n$, which is why we add the subscript $i$ to the design parameter in order to emphasize that they can be different. Consequently, the parameter can be chosen in an online fashion, specifically at each time $t_i^\off$. Similarly, the function~$S$ on the interval $[t_i^\off,t_i^\on)$ can be defined at each time $t_i^\off$. Particularly, $S$ simply needs to be picked so that $S(t_i^\off)>V(x(t_i^\off))$. 
\begin{remark}\longthmtitle{Generalization of the controller on trigger design}
In this paper, the trigger~\eqref{trigger:on} is given with a design parameter $c_{\beta,i}$. This parameter relates the relationship between the speed at which the certificate function~$V$ is allowed to approach the bound~$S$, and the difference, $S-V$, between the two. In general, this constant can be replaced with a class-$\mathcal K_\infty$ function $\map{\beta_i}{\real}{\real}$, and the proof for Lemma~\ref{lem:on} will invoke Nagumo's theorem~\cite{FB-SM:07}, instead of proving the positivity of the difference $S-V$ with the Comparison Lemma. The reason we use a linear function $\beta_i$ with the constant $c_{\beta_i}$ is that the exponential bound for $S-V$ may be useful for analyses in future works. We further believe it is more intuitive to tune a linear constant than a function.~\hfill$\bullet$
\end{remark}

\subsection{Trigger Scheme for Intermittent Implementation}

We now combine the two trigger conditions to obtain result on intermittent control. 

\begin{theorem}\longthmtitle{Trigger Scheme for Intermittent Implementation}
\label{thm:intermittent}
Consider the intermittently-implemented control system~\eqref{eq:intermittent} with the sequences  $\{t_i^\on\}_{i\in\naturals}$ and  $\{t_i^\off\}_{i\in\naturals}$ iteratively determined by trigger designs~\eqref{trigger:off} and \eqref{trigger:on}. In addition to all the assumptions of Lemmas~\ref{lem:miet} and~\ref{lem:on}, let $S$ be a bounded function such that $S(t)\geq V(x(t))$ for time $t\in[t_i^\on,t_i^\off)$ for all $i\in\naturals$. Then the safe performance set $\Cc$ is forward invariant, i.e., $x(t)\in\Cc(t)$ at all time.
\end{theorem}

\begin{proof}
From Lemma~\ref{lem:on} and the assumption of the theorem, we ensure that $S(t)\geq V(x(t))$ at time $[t_i^\on,t_{i+1}^\on)$. Consequently, let $S_{\max}\in \realpos$ be the upper bound of $S$, the compact sublevel set $\Omega = \setdefb{x\in\real^n}{V(x)\leq S_{\max}}$ is forward invariant, and we can use Lemma~\ref{lem:miet} to establish a MIET~$\tau$ to lower bound $t_{i+1}^\on-t_i^\on> \tau$. This rules out the possibility of Zeno behavior, and therefore, $\lim_{i\rightarrow\infty}t_i = \infty$. Hence, $x(t)\in \Cc(t)$ for all time, concluding the proof.
\end{proof}
\vspace{.2cm}

Theorem ~\ref{thm:intermittent} offers a trigger scheme that enables an intermittent implementation of the controller~$k$. Note importantly that there is an additional assumption introduced by the theorem on the performance specification function~$S$. That is, the theorem requires $S(t)\geq V(x(t))$ for time $t\in[t_i^\on,t_i^\off)$, However, because the function $S$ does not appear in the trigger design~\eqref{trigger:off} on the interval $[t_i^\on,t_i^\off)$, it can be implicitly defined to meet the requirement. Consequently, we only need to specify $S$ so that $S(t_i^\off)>V(x(t_i^\off))$. Nevertheless, if we want to achieve stabilization tasks, we need to further specify the convergence of~$S$. We formalize this in the next result.

\begin{corollary}\longthmtitle{Sublevel Set Stabilization}\label{cor:stabilize}
For the intermittently-implemented control system~\eqref{eq:intermittent} with the sequences  $\{t_i^\on\}_{i\in\naturals}$ and  $\{t_i^\off\}_{i\in\naturals}$ iteratively determined by trigger designs~\eqref{trigger:off} and \eqref{trigger:on}. Let the performance specification function~$S$ satisfy the assumptions of Theorem~\ref{thm:intermittent}. If $\lim_{t\rightarrow\infty}S$ exists then the sublevel set $\Cc_\infty = \setdefb{x\in\real^n}{V(x)\leq \lim_{t\rightarrow\infty} S(t)}$ is globally asymptotically stable. Furthermore, if $\lim_{t\rightarrow\infty}S=0$, then the origin of the system is globally asymptotically stable.~\hfill$\blacksquare$
\end{corollary}

Corollary~\ref{cor:stabilize} follows directly from the forward invariance of $\Cc$ from Theorem~\ref{thm:intermittent}, and it suggests an additional condition on $S$ in order to enforce the certificate $V$ to evolve in a meaningful way. As an example for how to prescribe the function $S$ so that the origin is rendered globally asymptotically stable, $S$ can be defined, for each interval $[t_i^\off,t_{i+1}^\on)$, with $S(t_i^\off)= (V(t_i^\on)+V(t_i^\off))/2$ and $\dot S = -\lambda S$ for some convergence rate $\lambda>0$. An alternative way of picking the function $S$ is to consider it as a performance criteria, cf.~\cite{PO-JC:21-tac}; however, such a method may require an estimation of a convergence parameter which may lead to conservatism in the overall convergence speed.

\begin{remark}\longthmtitle{Inactive Dwell Time}
{\rm
We refer to the length of time at which the controller can remain off as \textit{inactive dwell time}. The inactive dwell time is particularly important because it ties directly to resource conservation. In order to maximize the inactive dwell time, the function $S$ should be picked as large as possible. Such a choice of~$S$ would provide the certificate~$V$ more room to increase before the trigger condition~\eqref{trigger:on} is met. Note however that a large $S$ will slow down convergence towards the desirable set (or the origin). Another method to lengthen the inactive dwell time is by increasing $c_{\beta,i}$, which affects how close the value of $V$ can get to $S$ before the trigger condition is met. However, there is a limit to how much the inactive dwell time is lengthened via this method. That is, when $c_{\beta,i}\rightarrow \infty$, the inactive dwell time is precisely how long it takes for the value of the certificate $V$ to reach the function $S$.
\hfill~$\bullet$
}
\end{remark}

\section{Extension to Safety Constraints}

Thus far, we have synthesized intermittent controllers through the use of an ISS Lyapunov function, $V$, together with a performance specification function, $S$.  In particular, the main result (specifically Theorem \ref{thm:intermittent}) establishes that the trigger laws render the performance set $\Cc(t)$ forward invariant.  The goal of this section is to generalize these results to performance sets that may not be defined by Lyapunov level sets.   We consider a time-varying safe set: 
$$
\Hc(t) \defeq \setdefb{x\in \real^n}{h(x,t) \geq 0}
$$
where $h : \real^n \times \real \to \real$ is a continuously differentiable function, wherein its positivity defines a safety condition.  That is, the system is safe if $x(t) \in \Hc(t)$.  Note that the performance set, $\Cc(t) = \Hc(t)$ for $h(x,t) = S(t) - V(x)$. 

Our starting point is the assumption of safety robustness of a state-feedback system with a controller $k$, now designed to accomplish the task of guaranteeing safety. The assumption we use is the \emph{strong Input-to-State Safety (sISSf)} condition (see also,~\cite{AJT-PO-JC-AA:21-csl}): 
\begin{equation}\label{eq:sISSf}
 \underbrace{\frac{\partial h}{\partial t} + \left. \frac{\partial h}{\partial x} \right|_{(x,t)} f(x,k(x + e))}_{\defeq \dot{h}(x,e,t)}\geq -\omega(h(x,t))-\iota(\|e(t)\|)+d
\end{equation}
where $\omega$ and $\iota$ are class-$\Kc_\infty$ functions, and $d > 0$ is a constant. This condition suggests that the controller $k$ enforces forward invariance and asymptotic stability of a superlevel set~$\Hc_e$ smaller than~$\Hc$. The set $\Hc_e$ grows as $e$ grows in size.

In the stabilization case, we choose to turn off the controller when the certificate decreases too slowly, due the presence of sample-and-hold error. Analogously in the safety case, our idea is to turn off the controller when $\Hc_e$ grows too large (still contained within $\Hc$). Our trigger condition for turning off the safeguarding controller is given by:
\begin{subequations}\label{trigger:safe_off}
\begin{align}
t_i^* &\defeq \min\setdef{t>t_i^\on}{\dot h(x(t),e_i^\on(t),t) \nonumber\\ &\qquad \qquad \qquad \qquad= -\omega(h(x(t),t))+\theta d}\label{trigger:safe_off_deriv}\\
t_i^\off &\defeq \min\{t_i^*,t_i^\on+T_{\max}\}
\end{align}
\end{subequations}
where $\theta\in(0,1)$. As will be shown later, using this trigger design, $x(t)$ will be in the interior of $\Hc(t)$. Not only does this accomplish our goal of establishing safety, being in the interior of the set allow the possibility of turning off the controller because it will take at least some time to reach the boundary of the safe set.  We provide the following trigger design for when to turn the controller back on:
\begin{equation}\label{trigger:safe_on}
t_i^\on \defeq \small{\min\setdef{t>t_i^\off}{\dot h(x(t),e_i^\off(t),t) = -c_{\beta,i}\omega(h(x(t),t))} }   
\end{equation}
where $c_{\beta,i}$ is large enough such that:
$$
c_{\beta,i} >  -\frac{\dot h(x(t_i^\off),e_i^\off(t_i^\off),t_i^\off)}{\omega(h(x(t_i^\off),t_i^\off))}.
$$
This trigger design monitors the safety criterion and determines the time for turning the controller back on as soon as the criterion is violated. Assembling the two trigger designs above, we guarantee safety of system trajectories with the following result.

\begin{theorem}\longthmtitle{Trigger Scheme for Intermittent Implementation - Safeguarding Controller}\label{trigger:safe}
Consider the intermittently-implemented control system~\eqref{eq:intermittent}. Given a barrier function $h$ with the sISSf property~\eqref{eq:sISSf}, let the sequences  $\{t_i^\on\}_{i\in\naturals}$ and  $\{t_i^\off\}_{i\in\naturals}$ iteratively determined by trigger designs~\eqref{trigger:safe_off} and \eqref{trigger:safe_on}. If the function $f$ is bounded and  $\iota$ is Lipschitz continuous, then $\Hc$ is forward invariant, i.e., $x(t)\in\Hc(t)$ for all time.
\end{theorem}
\begin{proof}
We first note that according to~\cite[Thm. 1]{AJT-PO-JC-AA:21-csl}, there exists a MIET $\tau^*$ for the trigger design~\eqref{trigger:safe_off_deriv}, and thus, $t_i^\off-t_i^\on\geq \min\{\tau^*,T_{\max}\}$. Hence, Zeno behavior is ruled out and all system trajectories will be a complete solution.

At $t_i^\on$, the sample-and-hold error $e_i^\on(t_i^\on)=0$ by definition, so we have from~\eqref{eq:sISSf} that:
$$
\dot h(x(t),e(t),t) > -\omega(h(x(t),t))+\theta d.
$$
The trigger~\eqref{trigger:safe_off} then enforces the above inequality for the duration $[t_i^\on,t_i^\off)$ by continuity of the system trajectory. Consequently during this time duration, if $h(x) <\omega^{-1}(\theta d)$ ($\omega$ is invertible because it is a class-$\Kc$ function), then $\dot h > 0$. In particular, $\dot h > 0$ when $h(x)=0$, and thus according to Nagumo theorem~\cite{FB-SM:07}, $\Hc$ is forward invariant for the time interval $[t_i^\on,t_i^\off)$.

Also, following the deduction above, the function $h$ must evolve during $[t_i^\on,t_i^\off)$ from $h(x(t_i^\on))\geq 0$ to $h(x(t_i^\off))>0$. That is, at time $t_i^\off$, the inequality must be strict. This ensures that the trigger design~\eqref{trigger:safe_on} is well-defined with a possible choice of $c_{\beta,i}$. The bound on $c_{\beta,i}$ directly assures:
$$
\dot h(x(t),e(t),t) > -c_{\beta,i}\omega(h(x(t),t))
$$
at time $t_i^\off$. Then with continuity, the trigger design~\eqref{trigger:safe_on} ensures the inequality holds for the duration $[t_i^\off,t_{i+1}^\on)$. Once again, according to Nagumo theorem~\cite{FB-SM:07}, $\Hc$ is forward invariant for the time interval, concluding the proof.
\end{proof}

Theorem~\ref{trigger:safe} assembles the two trigger designs for turning the controller on and off based on a given barrier function. Our final result shows that our trigger scheme render the desired time-varying set forward invariant. The idea is analogous to the stabilization case. That is, we turn leave the controller on for as long as the controller drives the state away  from the boundary of the safe set. Then after we turn off the controller, we leave it off as long as it is safe.

\section{Application in Spacecraft Orbit Stabilization}

To illustrate the effectiveness of our trigger design, we study the task of stabilizing a satellite to a desired orbit. More specifically, we wish to find a discrete scheduling of satellite thrusters to maneuver the spacecraft around a central body so that it eventually tracks a circular orbit. We consider a satellite whose dynamics are described by Newton's gravitational model defined below, cf. \cite{WK-AS-SB:01}:
\begin{eqnarray}
\label{eqn:spacecraft}
\underbrace{ \begin{bmatrix}\dot r \\ \dot \theta \\ \dot z \\ \ddot r \\ \ddot \theta \\ \ddot z  \end{bmatrix} }_{\dot{x}} =  \underbrace{\begin{bmatrix}\dot r \\ \dot \theta \\ \dot z \\ r\dot \theta^2 - \mu r/(r^2+z^2)^{3/2} \\ -(2/r)\dot r \dot \theta \\ -\mu z/(r^2+z^2)^{3/2}  \end{bmatrix}}_{F(x)}+ \underbrace{\begin{bmatrix}
 0 \\ 0 \\ 0 \\ u_1 \\ u_2/r \\ u_3
 \end{bmatrix}}_{G(x) u} 
\end{eqnarray}
where $q = (r, \theta, z)$ are the cylindrical coordinates describing the radial position, angle, and height, respectively, of the satellite with respect to a desired orbital plane and the state $x = (q,\dot{q})$. Next, we illustrate our approach of designing an intermittent controller that stabilizes a desired orbit.

\begin{remark}\longthmtitle{Generalization to Feedback Linearization Problems} Even though we use the spacecraft orbit stabilization, as an example, to demonstrate our results, the general ideas in the following procedure are applicable to any other feedback linearizable system.~\hfill$\bullet$
\end{remark}
\subsection*{ \text{Task I:} Obtaining a Control Lyapunov Function}

We begin by using feedback linearization to help construct a Control Lyapunov Function (CLF). Feedback linearization allows us to synthesize a coordinate transformation from the nonlinear system \eqref{eqn:spacecraft} to a linear system, allowing us to use linear control techniques to synthesize a CLF. In particular, we consider (vector relative degree 2 \cite{SS:13}) outputs:
$$
y \defeq  \begin{bmatrix}
r-r_\des,  \\ \theta - (\theta_0+\sqrt{\mu/r_\des^3} t) \\  z 
\end{bmatrix} \quad \Rightarrow \quad 
\dot{y}  =  \begin{bmatrix}
 \dot r \\ \dot \theta - \sqrt{\mu/r_\des^3}  \\ \dot z
\end{bmatrix}
$$
where the origin corresponds to a circular orbit with radius $r_\des$.
It follows that $y$ has relative degree 2 since: 
\begin{align}
\label{eq: output dynamics}
\ddot{y}  &= \underbrace{\begin{bmatrix} r\dot \theta^2 - \mu r/(r^2+z^2)^{3/2} \\ -(2/r)\dot r \dot \theta \\ -\mu z/(r^2+z^2)^{3/2}  \end{bmatrix}}_{\Lie_F^2 y(x)} + \underbrace{\begin{bmatrix} 1 & 0 & 0 \\
0 & \frac{1}{r} & 0  \\ 0 & 0 & 1 \end{bmatrix}}_{\Lie_F \Lie_G y(x) } u 
\end{align}
with the decoupling matrix $\Lie_F \Lie_G y(x) $ invertible (we assume $r>0$ along the trajectory). 
Therefore, for $\eta = (y,\dot{y})$, applying the feedback linearizing control input: 
\begin{eqnarray}
\label{eqn:feedbackline}
u = \Lie_F \Lie_G y(x)^{-1}(- \Lie_F^2 y (x)  +v),
\end{eqnarray}
with an auxiliary input $v \in \real^3$, yields the linear system: 
$$
\dot \eta = \underbrace{\begin{bmatrix}
\mathbf 0_{3\times3} & \mathbf{I}_{3\times3} \\ \mathbf 0_{3\times3} & \mathbf 0_{3\times3}
\end{bmatrix}}_{A} \eta + \underbrace{\begin{bmatrix}\mathbf 0_{3\times3} \\ \mathbf{I}_{3\times3}\end{bmatrix}}_{B}v .
$$
Since $(A,B)$ is a controllable pair, the idea is to find a CLF with a gain matrix $K$ such that $(A+BK)$ is Hurwitz resulting in a closed-loop system:
$$
\dot{\eta} = (A+BK)\eta,
$$ 
which renders $\eta = 0$ (i.e. our desired circular orbit with radius $r_{des}$) exponentially stable.  
Particularly, there exists a Lyapunov matrix $P = P^{\top} \succ 0$ satisfying the Continuous-Time Lyapunov Equation (CTLE):
\begin{align} 
\label{eq:CTLE}
   (A+BK)^{\top} P + P(A+BK) = -Q 
\end{align}
for any given $Q = Q^{\top} \succ 0$. The above steps, through feedback linearization, help us synthesize a CLF which certifies stability of the origin of the system. 
More specifically:
\begin{align*}
&V_\fl(\eta) \defeq \eta^{\top} P \eta\\
&\dot{V}_\fl(\eta,v) \defeq \eta^{\top}(A^{\top}P +PA)\eta + 2\eta^{\top}PB v  
\end{align*}
In particular, there exists a control input $v = K \eta$ such that $\dot V_\fl(\eta,K\eta)=-\eta^\top Q \eta <0$, so $V$ is a valid CLF.

It is easy to show that this is also a CLF for the original nonlinear system \eqref{eqn:spacecraft}. In fact, we define the same Lyapunov function by using our coordinate transformation as:
\begin{align}
&V(x,t) \defeq \eta(x,t)^\top P\eta(x,t)\\
&\dot V(x,u,t) \defeq - \eta(x,t)^\top (A^{\top}P +PA)\eta(x,t)\nonumber\\
&\qquad \qquad ~+ 2\eta(x,t)^\top PB\big(\Lie_F^2 y (x) + \Lie_F \Lie_G y(x) u\big). 
\end{align}
Then we know for each $(x,t)$, there exists a 
$$
u^* = \Lie_F \Lie_G y(x)^{-1}(- \Lie_F^2 y (x)  +K\eta(x,t))
$$
such that $\dot V(x,u^*,t) = -\eta(x,t)^\top Q\eta(x,t)$. Although we could use this $u^*$, we follow the theme of this paper of conserving resource and instead use a pointwise optimal controller, as described in the next task.

\subsection*{ \text{Task II:} Controller Design}

Now that we have obtained a CLF, in order design a controller $\map{k}{\real^6}{\real^3}$, we consider a set of admissible inputs for the transformed system:
$$
K_u(x,t) \defeq \setdefb{u\in\real^m}{ \dot{V}(x,u,t) \leq -\eta(x,t)^{\top} Q \eta(x,t) }.
$$
Since we have a feasible constraint, we choose the control pointwise optimally with the following Quadratic Program (QP) based controller:
\begin{align*}\label{eq:CLF-QP}
k(x,t) &= \argmin_{u \in \mathbb{R}^3}||u||^2 \tag{CLF-QP}\\
&\text{s.t.}\ \dot{V}(x,u,t) \leq - \eta^{\top} Q \eta
\end{align*}
Because the controller must satisfy its constraint, the control feedback $u = k(x)$ renders the origin exponentially stable. This is the controller we will be using for our system.

\subsection*{ \text{Task III:} Establishing ISS Lyapunov Function}
Now that we have obtained a controller \eqref{eq:CLF-QP}, we need to establish that the closed-loop system has an ISS Lyapunov function \eqref{eq:iss}, in order to apply our results to implement the controller in an intermittent fashion. Given the sample and hold implementation $u = k(x(t_i))$, we begin by transforming \eqref{eq:sah} to the linear case with:
$$
v=\Lie_F^2 y (x) + \Lie_F \Lie_G y(x) k(x(t_i))
$$
where $k(x(t_i))$ is the control input being held constant at $t_i$.\\
Our goal to establish an ISS Lyapunov Function of the form:
\begin{align}
\label{eq: ISS-linear}
    \left.\frac{dV}{d\eta}\right|_\eta&(A\eta+Bv) \leq-\alpha(\|\eta\|)+\gamma(\|e_i\|) 
\end{align}
where $e_i = \eta(t_i)-\eta$ is the error introduced from holding our control input. To establish \eqref{eq: ISS-linear}, we analyse how $v$ (from the held control input) deviates from  the desired auxiliary input $v_\des$, which is given by: 
$$
v_\des=\Lie_F^2 y (x) + \Lie_F \Lie_G y(x) k(x)
$$
Given $v$, $v_\des$ and $k(x(t_i))$, we can explicitly write the evolution of the Lyapunov function along the trajectory as:
\begin{align*}
    \left.\frac{dV}{d\eta}\right|_\eta&(A\eta+Bv) = \eta^\top(A^\top P+PA)\eta+2\eta^\top PBv\\
    &= \eta^\top(A^\top P+PA)\eta+2\eta^\top PBv_\des \\
        & \qquad+ 2\eta^\top PB(v-v_\des) \\
    &\leq -\eta^\top Q\eta +2\eta^\top PB(v-v_\des)\\
    & = -\eta^\top Q\eta +2\eta^\top PB \Lie_F \Lie_G y(x)(k(x(t_i))-k(x))
\end{align*}

The above equation is not yet in the form described in \eqref{eq: ISS-linear}. In order to proceed with our derivation, we will make a few assumptions.
\begin{assumption}\longthmtitle{Bounded Decoupling Matrix Assumption}
Let $t\mapsto x(t)$ be the solution to the system \eqref{eqn:spacecraft}. We assume $\|\Lie_F \Lie_G y(x(t)\| \leq H$ $\forall$ t.~\hfill$\bullet$
\end{assumption}
For our spacecraft problem, the assumption is reasonable because we will consider only initial conditions such that we can guarantee $r > r_{\min}$ along the trajectory for some $r_{\min}>0$ corresponding to the minimum radius for the satellite to crash into the central body. This allows us to assume that $\Lie_F \Lie_G y(x)$ described in \eqref{eq: output dynamics} is bounded above. 

Given this assumption and noting that $k$ is Lipschitz continuous because it is a QP-based controller with one control affine constraint~\cite{XX-PT-JWG-ADA:15}, we obtain:
$$
\left.\frac{dV}{d\eta}\right|_\eta(A\eta+Bv)\leq -\eta^\top Q\eta +2\eta^\top M\|e_{x,i}\|
$$
where $M = \|PB\|HL_k$ for some Lipschitz constant $L_k$ and $e_{x,i}=x(t_i)-x$. This shows that $V$ is an ISS Lyapunov function with respect to the state deviation error in the original nonlinear coordinate. To get the error in the linear coordinate, we further make the ensuing assumption.

\begin{assumption}\longthmtitle{Relationship Between State Deviations}
There exists a class~$\mathcal K$ function $\psi$ such that the state deviations
of the two different coordinates are related as $\|e_{x,i}\|\leq \psi(\|e_i\|)$
where $e_i = \eta(t_i)-\eta$.~\hfill$\bullet$
\end{assumption} 
This assumption relies on the intuition that the coordinate transformation is diffeomorphic. However, for our problem, the transformation is also time dependent, so it is unclear whether this will be true and is part of our ongoing research.

Finally, with the above assumption:
\begin{align*}
    \left.\frac{d V}{d\eta}\right|_\eta(A\eta+B v) &\leq -\eta^\top Q\eta + 2\eta^\top M\|e_{x,i}\| \\
    & \leq -\lambda_{min}(Q)\|\eta\|^2 + \|2\eta^\top M\| \psi(\|e_{i}\|)\\
    & \leq -\alpha(\|\eta\|)+\gamma(\|e_i\|) 
\end{align*}
where $\alpha$ and $\gamma$ can be found using Young's Inequality \cite{GHH-JEL-GP:52}.
\subsection*{ \text{Task IV:} Intermittent-Control Trigger Design}
Now that we have obtained an ISS Lyapunov Function and the controller \eqref{eq:CLF-QP}, we will apply the results in this paper to implement the controller in an intermittent fashion. 

Particularly, we design a trigger law that corresponds to (\ref{trigger:off}) for triggering the controller off, as follows:
\begin{align*}
  t_i^* = \min\setdefb{t>t_i^\on}{\small{\left.\frac{dV}{d\eta}\right|_\eta}(A\eta+Bv)+(1-\sigma)\eta^\top Q\eta=0}
\end{align*}
Then, once we turn the controller off at $t_i^\off$, we need to prescribe conditions required to find a time to turn the controller back on. To do this, we need to define a \textit{safe performance set} \eqref{eqn:setCc} described by a \textit{performance specification function} $S$. With $\lambda>0$, let $S$ be:
$$
S(t_i^\off)= (V(t_i^\on)+V(t_i^\off))/2,~\dot S = -\lambda S.
$$ 
for time $t\in[t_i^\off,t_{i+1}^\on)$. This helps us set the conditions to turn the controller back on. More specifically, by choosing the design parameter $c_{\beta,i}>0$ satisfying \eqref{eq:beta_param}, the resulting trigger law will be of the form described in \eqref{trigger:on}:
\begin{multline*}
t_{i+1}^\on \defeq \min\setdefb{t>t_i^\off}
{\left.\frac{dV}{d\eta}\right|_\eta(A\eta+Bv)-\dot S(t)\\=c_{\beta,i}(S(t)-V(\eta))}.
\end{multline*}

We simulate our results with an eighth-order harmonics gravity model in MATLAB.
Due to imperfection in the shape of the asteroid, there are gravity disturbances 
unaccounted for by the Newton's gravitational model. Nevertheless, our results illustrate robustness of our controller design. Figure~\ref{fig: trajectory} shows the resulting trajectory. At all times, the spacecraft stays within the safe performance set (transparent red). Note that the figure shows the safe performance set at the final time, so past trajectory may appear out of bounds. In addition, as shown in Figure~\ref{fig: controller} (top),  we report that the certificate~$V$ (blue) remains below the performance specification~$S$ (red) at all times. With the thruster firing limit time of $T_{\max}$ of 10 seconds, the control inputs appears as spikes in Figure~\ref{fig: controller} (middle), and the intermittent nature of the control input is evident. Finally, Figure~\ref{fig: controller} (bottom) verifies that the radial distance (blue) remains within the prescribed safe performance set~$\Cc$ (red) at all time.

\begin{figure}[t!]
\centering
\includegraphics[width=\linewidth,height=\textheight,keepaspectratio]{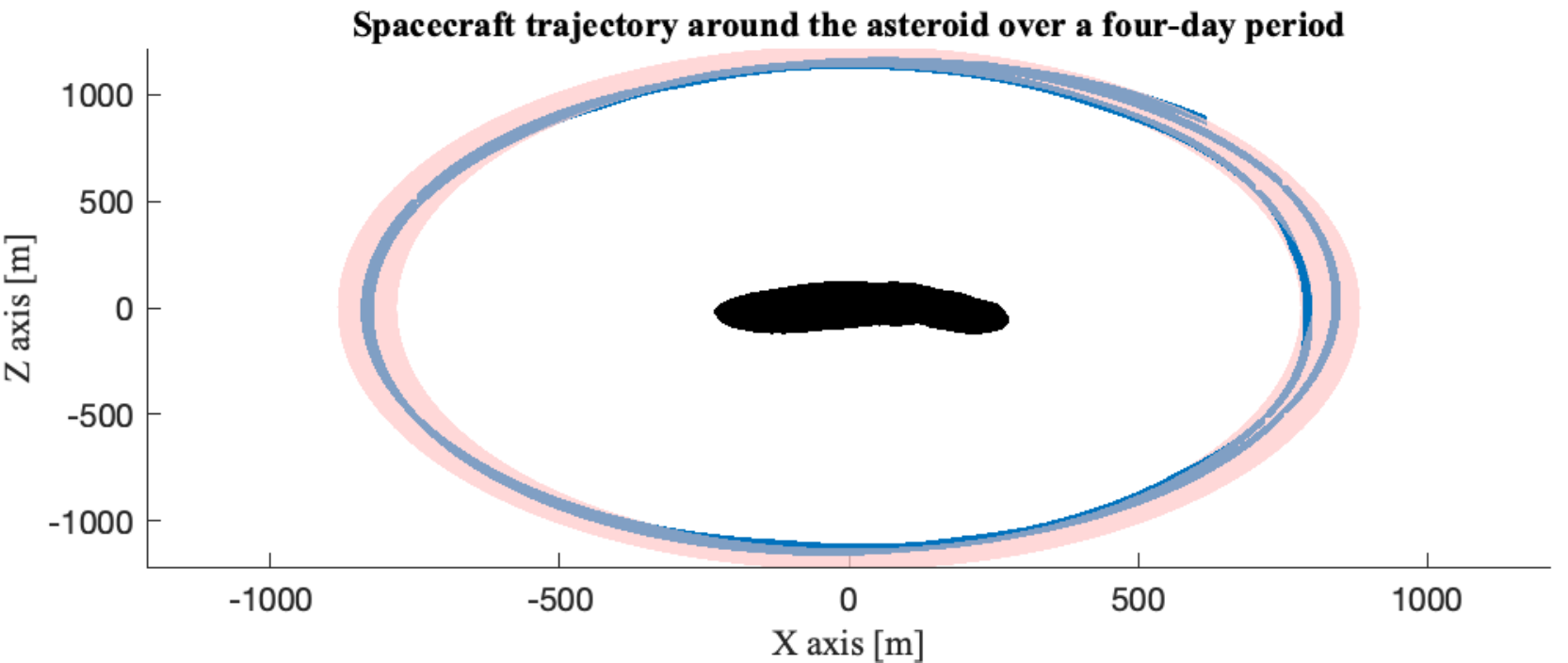}
\caption{The figure shows the trajectory as projected on the x-z plane, which is perpendicular to the ecliptic plane. The desired orbit has right ascension of the ascending node (RAAN) of $45^\circ$ and an inclination angle of $90^\circ$.}

\label{fig: trajectory}
\end{figure}

\begin{figure}[t!]
\centering
\includegraphics[width=0.98\linewidth,height=\textheight,keepaspectratio]{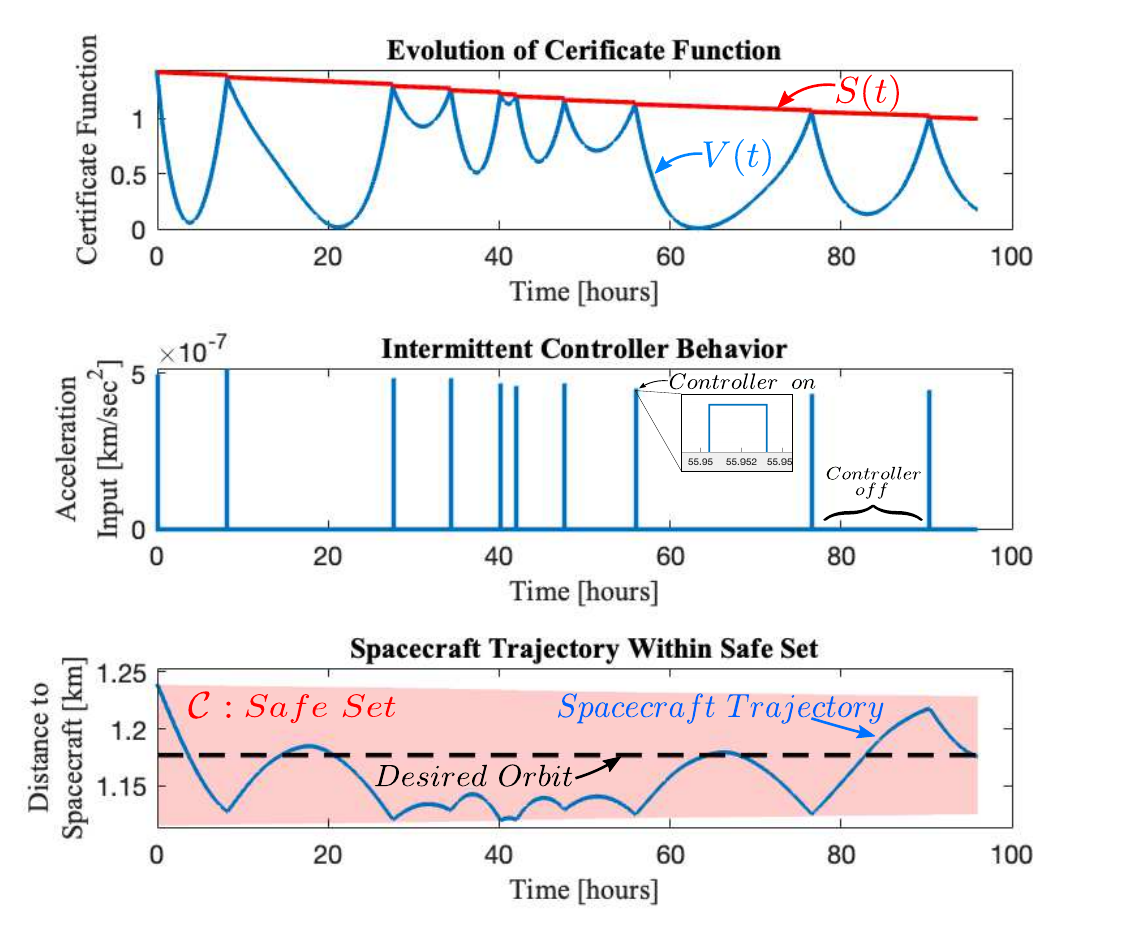}
\caption{(Top) Certificate function along the trajectory. (Middle) Acceleration input. Each spike lasts no longer than $T_{\max}$ of 10 seconds. (Bottom) Distance of the spacecraft from the central body with estimated safe performance set $\Cc$ getting smaller over time.}
\label{fig: controller}
\end{figure}

\section{Conclusion}
This paper synthesized trigger schemes that permit the intermittent implementation of state-feedback controllers
in both the context of stabilization and safety. In particular, our trigger scheme turns the controller off when it no longer provides satisfactory performance, and turns the controller back on before all the prior progress made by the controller is undone. To demonstrate the effectiveness of our results, we applied them to the problem of spacecraft orbit stabilization. Future work will study the co-design problem, seeking controllers that can lengthen the controller inactive dwell time for as long as possible.  This will be studied in the context of stability, safety, and their interaction.

\vspace{0.2cm}
\noindent \textbf{Acknowledgement.}
The authors would like to thank JPL for their feedback on the application of these ideas to spacecraft, and Saptarshi Bandyopadhyay in particular for discussions and providing the  eighth-order harmonics gravity model used in our simulation results.

\bibliography{alias,PO,main-Pio}

\begin{thebibliography}{10}

\bibitem{PT:07}
P.~Tabuada, ``Event-triggered real-time scheduling of stabilizing control
  tasks,'' {\em IEEE Transactions on Automatic Control}, vol.~52, no.~9,
  pp.~1680--1685, 2007.

\bibitem{WPMHH-KHJ-PT:12}
W.~P. M.~H. Heemels, K.~H. Johansson, and P.~Tabuada, ``An introduction to
  event-triggered and self-triggered control,'' in {\em {IEEE} Conf.\ on
  Decision and Control}, (Maui, HI), pp.~3270--3285, Dec. 2012.

\bibitem{RP-PT-DN-AA:15}
R.~Postoyan, P.~Tabuada, D.~Ne\v{s}i\'{c}, and A.~Anta, ``A framework for the
  event-triggered stabilization of nonlinear systems,'' {\em IEEE Transactions
  on Automatic Control}, vol.~60, no.~4, pp.~982--996, 2015.

\bibitem{CL-GF-XL:07}
C.~Li, G.~Feng, and X.~Liao, ``Stabilization of nonlinear systems via
  periodically intermittent control,'' {\em IEEE Transactions on Circuits and
  Systems II: Express Briefs}, vol.~54, no.~11, pp.~1019--1023, 2007.

\bibitem{NY-YS-KK-TN:16}
N.~Yoshikawa, Y.~Suzuki, K.~Kiyono, and T.~Nomura, ``Intermittent
  feedback-control strategy for stabilizing inverted pendulum on manually
  controlled cart as analogy to human stick balancing,'' {\em Frontiers in
  Computational Neuroscience}, vol.~10, 2016.

\bibitem{RD-HP:19}
R.~Dash and H.~J. Palanthandalam-Madapusi, ``When to use intermittent control
  for stabilization?,'' in {\em 2019 Sixth Indian Control Conference (ICC)},
  (IIT Hyderabad, India), pp.~526--531, Dec. 2019.

\bibitem{QW-YH-GT-MW:16}
Q.~Wang, Y.~He, G.~Tan, and M.~Wu, ``Stabilization of linear systems via
  state-dependent intermittent control,'' in {\em IEEE Chinese Control
  Conference (CCC)}, pp.~1556--1561, 2016.

\bibitem{LY-WM-IJ-PE:93}
L.~Yang, W.~Marian, I.~James, and P.~Ekaterina, ``Intermittent control of
  coexisting attractors,'' {\em Philosophical Transactions of the Royal Society
  A: Mathematical, Physical and Engineering Sciences}, vol.~371, no.~1993,
  p.~20120428, 2013.

\bibitem{PJG-LW:09}
P.~J. Gawthrop and L.~Wang, ``Event-driven intermittent control,'' {\em
  International Journal of Control}, vol.~82, no.~12, pp.~2235--2248, 2009.

\bibitem{PJG-IL-ML-HG:11}
P.~J. Gawthrop, I.~Loram, M.~Lakie, and H.~Gollee, ``Intermittent control: a
  computational theory of human control,'' {\em Biological cybernetics},
  vol.~104, no.~1, pp.~31--51, 2011.

\bibitem{PJG-IL-HG-ML:14}
P.~J. Gawthrop, I.~Loram, H.~Gollee, and M.~Lakie, ``Intermittent control
  models of human standing: similarities and differences,'' {\em Biological
  cybernetics}, vol.~108, no.~2, pp.~159--168, 2014.

\bibitem{EDS:08}
E.~D. Sontag, ``Input to state stability: Basic concepts and results,'' {\em
  Nonlinear and Optimal Control Theory}, vol.~1932, pp.~163--220, 2008.

\bibitem{PJG-LW:10}
P.~J. Gawthrop and L.~Wang, ``Intermittent redesign of continuous
  controllers,'' {\em International Journal of Control}, vol.~83, no.~8,
  pp.~1581--1594, 2010.

\bibitem{AG:15}
A.~Girard, ``Dynamic triggering mechanisms for event-triggered control,'' {\em
  IEEE Transactions on Automatic Control}, vol.~60, pp.~1992--1997, 2015.

\bibitem{PO-JC:21-tac}
P.~Ong and J.~Cort\'es, ``Performance-barrier-based event-triggered control
  with applications to network systems,'' {\em IEEE Transactions on Automatic
  Control}, 2021.
\newblock Submitted.

\bibitem{ADA-SC-ME-GN-KS-PT:19}
A.~D. Ames, S.~Coogan, M.~Egerstedt, G.~Notomista, K.~Sreenath, and P.~Tabuada,
  ``Control barrier functions: Theory and applications,'' in {\em {E}uropean
  {C}ontrol {C}onference}, (Naples, Italy), pp.~3420--3431, June 2019.

\bibitem{GY-CB-RT:19}
G.~Yang, C.~Belta, and R.~Tron, ``Self-triggered control for safety critical
  systems using control barrier functions,'' in {\em {A}merican {C}ontrol
  {C}onference}, (Philadelphia, PA), pp.~4454--4459, July 2019.

\bibitem{AJT-PO-JC-AA:21-csl}
A.~J. Taylor, P.~Ong, J.~Cort\'es, and A.~Ames, ``Safety-critical event
  triggered control via input-to-state safe barrier functions,'' {\em IEEE
  Control Systems Letters}, vol.~5, no.~3, pp.~749--754, 2021.

\bibitem{SK-ADA:18}
S.~Kolathaya and A.~D. Ames, ``Input-to-state safety with control barrier
  functions,'' {\em IEEE Control Systems Letters}, vol.~3, no.~1, pp.~108--113,
  2018.

\bibitem{ADA-XX-JWG-PT:17}
A.~D. Ames, X.~Xu, J.~W. Grizzle, and P.~Tabuada, ``Control barrier function
  based quadratic programs for safety critical systems,'' {\em IEEE
  Transactions on Automatic Control}, vol.~62, no.~8, pp.~3861--3876, 2017.

\bibitem{HKK:02}
H.~K. Khalil, {\em Nonlinear Systems}.
\newblock Prentice Hall, 3~ed., 2002.

\bibitem{FB-SM:07}
F.~Blanchini and S.~Miani, {\em Set-Theoretic Methods in Control}.
\newblock Boston, MA: Birkh{\"a}user, 2007.

\bibitem{WK-AS-SB:01}
W.~Kang, A.~Sparks, and S.~Banda, ``Coordinated control of multisatellite
  systems,'' {\em {AIAA} Journal of Guidance, Control, and Dynamics}, vol.~24,
  no.~2, pp.~360--368, 2001.

\bibitem{SS:13}
S.~Sastry, {\em Nonlinear Systems: Analysis, Stability, and Control}.
\newblock New York: Springer, 2013.

\bibitem{XX-PT-JWG-ADA:15}
X.~Xu, P.~Tabuada, J.~W. Grizzle, and A.~D. Ames, ``Robustness of control
  barrier functions for safety critical control,'' {\em IFAC-PapersOnLine},
  vol.~48, no.~27, pp.~54--61, 2015.

\bibitem{GHH-JEL-GP:52}
G.~H. Hardy, J.~E. Littlewood, and G.~Polya, {\em Inequalities}.
\newblock Cambridge, UK: Cambridge University Press, 1952.

\end{thebibliography}
\bibliographystyle{ieeetr}
\end{document}